\documentclass[11pt]{article}

\usepackage[T1]{fontenc}
\usepackage[utf8]{inputenc}
\usepackage{lmodern}
\usepackage[a4paper,margin=0.90in]{geometry}
\usepackage{amsmath,amssymb,amsfonts,mathtools,bm}
\usepackage{amsthm}
\usepackage{braket}
\usepackage{microtype}
\usepackage{booktabs}
\usepackage{tabularx}
\usepackage{array}
\usepackage{makecell}
\usepackage{enumitem}
\usepackage{float}
\usepackage{xcolor}
\usepackage{graphicx}
\usepackage{tikz}
\usetikzlibrary{arrows.meta,positioning,shapes.geometric,calc,fit}
\usepackage[hidelinks]{hyperref}
\hypersetup{
  pdftitle={Ordered-Angle Coding for Exact Multiuser Unanimity Testing},
  pdfauthor={Luis Adri\'an Lizama-P\'erez},
  pdfkeywords={quantum communication, ordered-angle coding, unanimity testing, Bell-state readout, transformation-only users, serial quantum networks, calibration robustness}
}
\usepackage[nameinlink,noabbrev]{cleveref}
\usepackage[numbers,sort&compress]{natbib}

\newcolumntype{Y}{>{\raggedright\arraybackslash}X}
\newcolumntype{P}[1]{>{\raggedright\arraybackslash}p{#1}}

\newcommand{\Tr}{\operatorname{Tr}}

\newcommand{\Ztwo}{\mathbb{Z}_2}
\newcommand{\doteqph}{\mathrel{\dot{=}}}
\newcommand{\ketb}[1]{\lvert\beta_{#1}\rangle}
\newcommand{\proj}[1]{\lvert#1\rangle\!\langle#1\rvert}

\newcommand{\nuTwo}{\nu_2}

\newtheorem{theorem}{Theorem}
\newtheorem{proposition}{Proposition}
\newtheorem{corollary}{Corollary}
\newtheorem{definition}{Definition}

\title{Ordered-Angle Coding for Exact Multiuser Unanimity Testing}
\author{%
Luis Adri\'an Lizama-P\'erez\\[3pt]
\small Departamento de Sistemas de Informaci\'on y Comunicaciones,\\
\small Divisi\'on de Ciencias B\'asicas e Ingenier\'ia, Universidad Aut\'onoma Metropolitana, Unidad Lerma,\\
\small Lerma, Estado de M\'exico 52005, Mexico\\
\small \texttt{l.lizama@correo.ler.uam.mx}}
\date{September 2026}

\begin{document}
\maketitle

\begin{abstract}
We introduce ordered-angle coding for binary-unanimity testing among $n$ transformation-only users in a serial quantum architecture inherited from Loop-Back communication. User $B_i$ fixes one private sign across a logical word and applies $R(s_i\alpha_j)$ in trial $j$. If $w$ users choose the negative sign, serial composition gives $R[(n-2w)\alpha_j]$. We show that every fixed same-axis pulse that is deterministic for both unanimous inputs under the binary Bell readout has $\alpha_j=q_j\pi/(2n)$ with integer $q_j$. For a nonadaptive word $\mathbf q=(q_1,\ldots,q_m)$, a mixed Hamming weight imitates the unanimous signature with probability
\[
M_{\mathbf q}(w)=\prod_{j=1}^{m}\cos^2\!\left(\frac{\pi q_jw}{n}\right).
\]
Within this complete deterministic-unanimity pulse family, a finite perfect word exists if and only if $n$ is a power of two. For $n=2^r$, the dyadic word $(1,2,4,\ldots,2^{r-1})$ is exact and pulse-minimal with $m_{\min}=r=\log_2n$ trials. It replaces the $O(n^2\log(1/\varepsilon))$ worst-case burden of repeated smallest-angle tests by exact $O(\log n)$ verification. Odd primes instead admit balanced statistical words with uniform mixed-weight imitation $2^{-(p-1)}$.

For actual operations $R(s_i\alpha_j+\delta_{ij})$, an ideal rejecting position is lifted to $\sin^2\Delta_j$, where $\Delta_j=\sum_i\delta_{ij}$. This yields explicit robustness bounds with angular and binary-readout errors. Bell entanglement is not required for the additive algebra, but it keeps the traveler locally maximally mixed in every honest trial. We therefore present the result as a coded multiuser relation primitive with optional raw conference-key-candidate use, not as a composably secure conference-key protocol.
\end{abstract}


\section{Introduction}

Quantum conference key agreement and multiparty quantum key agreement seek common secret material among several participants, while private-comparison protocols evaluate relations among distributed inputs under different trust and capability assumptions~\cite{murta2020review,lin2021mqka,mouaji2026review}. These tasks have been realized with Bell and GHZ states, single particles, interference, and device-independent resources. Their diversity makes it important to separate network architecture, quantum resource, relation functionality, and security model.

This work studies a serial quantum relation architecture inherited from Loop-Back communication. Alice prepares a quantum reference, one traveling qubit visits remote users in sequence, each user applies a private transformation, and Alice performs the final observation~\cite{lizama2025loop,lizama2025three}. Earlier Loop-Back papers called these endpoints ``passive'' because they require no protocol-carrier source, detector, measurement module, or quantum memory. Here we use the more precise term \emph{transformation-only} to avoid confusion with fully passive source architectures in contemporary QKD and CKA. The immediate predecessor established the two-user retained-half Bell geometry with same-axis operations $R(\pm\alpha)$ and an Alice-private binary Bell comparator~\cite{lizama2026twouser}.

The new question is not whether rotations add. It is whether a deliberately scheduled set of common rotation magnitudes can act as a finite relation code. If every user retains the same private sign across a word while the public magnitude changes between trials, each position probes a different Hamming-weight class of the collective input. For
\begin{equation}
\mathbf q=(q_1,\ldots,q_m),\qquad \alpha_j=\frac{q_j\pi}{2n},
\end{equation}
the complete binary record has a deterministic unanimous signature and a mixed Hamming weight $w$ imitates that signature with probability
\begin{equation}
M_{\mathbf q}(w)=\prod_{j=1}^{m}\cos^2\!\left(\frac{\pi q_jw}{n}\right).
\label{eq:intro-mimic}
\end{equation}
The angle family is not imposed arbitrarily. We prove that every fixed same-axis pulse that is deterministic for both unanimous inputs under the binary Bell instrument has exactly the form $\alpha=q\pi/(2n)$ with integer $q$. The codeword problem therefore concerns nonadaptive schedules drawn from the complete deterministic-unanimity same-axis pulse family.

Within that family, finite exact discrimination is possible precisely when $n$ is a power of two. For $n=2^r$, the dyadic word $\mathbf q=(1,2,4,\ldots,2^{r-1})$ is exact and pulse-minimal. Each position eliminates one $2$-adic valuation class of mixed Hamming weights, so $r=\log_2n$ Bell trials are necessary and sufficient. By contrast, repeated use of the smallest deterministic angle needs $O(n^2\log(1/\varepsilon))$ trials to drive the worst mixed imitation below a fixed target $\varepsilon$. User counts with odd factors form a statistical regime, with odd primes admitting balanced words that equalize residual mixed-weight ambiguity.

For a quantum-technology setting, exact algebra must also be tied to control tolerances. We therefore treat additive per-user angular offsets and binary readout errors. If the actual operation is $R(s_i\alpha_j+\delta_{ij})$, the total unwanted angular shift is $\Delta_j=\sum_i\delta_{ij}$. An ideal rejecting position is then lifted from zero to $\sin^2\Delta_j$. This yields explicit false-acceptance and true-acceptance bounds and a conservative per-user calibration requirement that tightens approximately as $1/n$ at fixed relation-error target.

Bell entanglement is deliberately separated from the codeword algebra. Same-axis rotations compose identically on a pure qubit, and pure references with $\langle Y\rangle=0$ reproduce the same ideal $\cos^2\Theta$ overlap law. Bell references are retained because the traveling subsystem is locally $I/2$ in every realized honest trial. The useful relation information is therefore stored in the correlation with Alice's retained subsystem rather than in a local Bloch direction. This is a resource-level information property, not a proof of composable secrecy against intermediate coherent access.

\subsection*{Main contributions}
\begin{enumerate}[leftmargin=2em]
\item We formulate ordered-angle coding as a multiuser relation primitive and derive the word-imitation law in equation~\eqref{eq:intro-mimic}.
\item We characterize the complete fixed same-axis pulse family that is deterministic under unanimity, then prove exactness if and only if $n=2^r$ for finite nonadaptive schedules from that family. We also prove the dyadic word pulse-minimal with $m_{\min}=r=\log_2n$.
\item We characterize the non-power-of-two regime through balanced odd-prime words and finite-budget statistical codewords.
\item We derive closed-form robustness bounds for angular calibration and binary readout error while keeping composable conference-key security and active intermediate attacks outside the claim.
\end{enumerate}

\section{Relation to prior work and scope of novelty}
\label{sec:prior}

\subsection{Multiparty key agreement and restricted-user architectures}

Bell-state, GHZ-state, single-particle, interference-based, and device-independent approaches to multiparty key establishment are established~\cite{shi2013bellmqka,liu2013singleparticles,ribeiro2018di,grasselli2019cka,zhao2020phase}. Experiments have demonstrated QCKA with multipartite entanglement and routed quantum networks~\cite{proietti2021experimental,pickston2023network}. More recent implementations address detector-side and source-side trust through measurement-device-independent and source-independent designs~\cite{yang2024mdiqcka,hua2025siqcka}. Security work also reaches finite-key coherent attacks and semi-quantum users with restricted capabilities~\cite{krawec2025cad,barreiro2025sqcka}. Fully passive CKA has been proposed in an interference-based prepare-and-measure model~\cite{li2026fullypassive}. These results establish that neither multiparty key generation nor limited endpoint capability is itself new.

The present architecture occupies a different design point. One Bell traveler visits the users serially, each remote node performs only a prescribed transformation, and Alice holds the reference and final measurement. An early Bell-network precedent already combined a retained Bell particle with local unitary operations~\cite{li2005network}. Banerjee \emph{et al.} accumulated Pauli operation sets on a Bell traveler for quantum conference functionality~\cite{banerjee2018conference}. Measurement-free mediated QKD likewise shows that removing user measurements does not by itself establish security against the mediator~\cite{zhou2024measurementfree}. Transformation-only operation is therefore an architectural constraint in this paper, not a security theorem.

\subsection{Rotation-based equality and quantum private comparison}

Quantum private comparison provides an established setting for learning equality while limiting disclosure of underlying inputs~\cite{yangwen2009,wu2023mqpc}. Rotation operations are also established in this literature. Hou and Wu use $R_y$ operations with Bell states for two-party equality testing~\cite{hou2025ry}, and a multiparty construction uses rotations before a semi-honest third party performs comparison~\cite{hou2025mqpcrotation}. Recent variants combine Bell resources with semi-quantum participants~\cite{zhou2025semibellqpc} and encode multiparty inputs on GHZ resources using local rotations~\cite{hou2025ghzrotation}. A broad claim such as ``rotation-based multiparty equality'' would therefore be untenable.

Our novelty claim is narrower. The object studied here is a position-labelled common-angle word generated by serial same-axis composition. The users keep one binary sign throughout a logical word, while the common magnitude changes between trials. This creates the product law $M_{\mathbf q}(w)$, lets different positions eliminate different Hamming-weight classes, and leads to the power-of-two exactness and pulse-minimality results. Equality or unanimity is the relation evaluated by the code, not the source of novelty.

\subsection{Hamming-weight computation with one-qubit rotations}

There is also direct prior art for Hamming-weight-dependent quantum computation with one writable qubit. Maslov \emph{et al.} showed that $n$-bit symmetric Boolean functions can be implemented exactly with a one-qubit workspace using quantum signal processing~\cite{maslov2021limitedspace}. Kawachi and Nishikubo subsequently developed private quantum signal processing in a simultaneous-message setting. Their evaluator acts on a function of the sum of parties' private angle parameters and their construction includes symmetric Boolean functions with one-qubit workspace~\cite{kawachi2025privateqsp}. These results establish that summed rotations and Hamming-weight-dependent one-qubit processing are not unique to Loop-Back.

The physical and coding models differ. The present scheme sends the same physical traveler sequentially through transformation-only users, does not use a writable evaluator or a QSP phase sequence, and uses a fresh Bell reference/complement test for every scheduled common angle. The design problem is to choose a finite set of common magnitudes that annihilates nonunanimous Hamming classes under that readout. The $2$-adic exactness and pulse-minimality results below are claims about this deterministic-unanimity Bell-readout family, not about arbitrary one-qubit algorithms for symmetric Boolean functions.

\subsection{Pulse sequences and the Loop-Back lineage}

Pulse-sequence equality tests also exist outside Loop-Back. Multi-party quantum fingerprinting uses phase-modulated coherent pulse trains and optical multiports, while more general fingerprinting networks optimize referee decisions for distributed relations~\cite{gomez2020fingerprinting,qin2021fingerprinting}. Those protocols use simultaneous-message interference. Here the physical object is a sequential Bell traveler and the logical information is encoded in a scheduled series of common rotation magnitudes.

Several carriers have also appeared earlier within Loop-Back communication~\cite{lizama2025loop,lizama2026twouser}. Table~\ref{tab:multipulse-lineage} summarizes the operational distinction. The contribution is not the use of multiple pulses. It is the replacement of repeated or filtered trials by complementary angular queries whose ordered signature constitutes the logical observable.

\begin{table}[H]
\centering
\caption{Different logical uses of multiple quantum carriers in the Loop-Back lineage.}
\label{tab:multipulse-lineage}
\small
\begin{tabularx}{\linewidth}{P{3.0cm}Y Y}
\toprule
Construction & Acceptance object & Role of additional carriers \\
\midrule
Earlier multipulse filter~\cite{lizama2025loop} & All returned BB84-basis carriers satisfy an orthogonality criterion & Redundancy used as an error filter \\
Two-user equal-angle Bell repetition~\cite{lizama2026twouser} & Repeated outcomes of one fixed angular question & Statistical amplification of one relation test \\
Ordered-angle codeword (Bell implementation) & Position-labelled signature for $\alpha_j=q_j\pi/(2n)$ & Different positions reject different mixed Hamming-weight classes \\
\bottomrule
\end{tabularx}
\end{table}

The immediate two-user predecessor supplies the retained Bell half, the $R(\pm\alpha)$ composition law, Alice's private reference, and the binary comparator~\cite{lizama2026twouser}. The present work begins when one input sign is held across several scheduled trials and the angle sequence itself is designed as a codeword.

\section{Bell-reference model and rotation convention}

Alice prepares a Bell state
\begin{equation}
\ketb r=(I\otimes P_r)\ket{\Phi^+},
\qquad
P_r=X^{r_x}Z^{r_z},
\qquad
r=(r_x,r_z)\in\Ztwo^2,
\label{eq:bellref}
\end{equation}
Modulo global phase, this convention gives
\[
P_{00}=I,\qquad P_{01}=Z,\qquad P_{10}=X,\qquad P_{11}=Y.
\]
where
\[
\ket{\Phi^+}=\frac{\ket{00}+\ket{11}}{\sqrt2}.
\]
The Bell label $r$ is sampled afresh and kept private by Alice in every trial. ``Private'' here means unavailable to the users and external observers. Alice necessarily knows the label. She retains subsystem $A$ and sends subsystem $B$ through the ordered route
\begin{equation}
A\longrightarrow B_1\longrightarrow B_2\longrightarrow\cdots
\longrightarrow B_n\longrightarrow A.
\label{eq:path}
\end{equation}
For every Bell label,
\begin{equation}
\rho_B=\Tr_A\proj{\beta_r}=\frac{I}{2}.
\label{eq:maxmixed}
\end{equation}

The native Loop-Back rotation is
\begin{equation}
R(\theta)=e^{-i\theta Y}
=
\begin{pmatrix}
\cos\theta&-\sin\theta\\
\sin\theta&\cos\theta
\end{pmatrix}
=R_y(2\theta),
\label{eq:R}
\end{equation}
so $\theta$ is the amplitude-space parameter and $2\theta$ is the conventional Bloch-sphere angle. Same-axis rotations compose additively,
\begin{equation}
R(\theta_2)R(\theta_1)=R(\theta_1+\theta_2).
\label{eq:add}
\end{equation}
For any Bell reference,
\begin{equation}
\left|\bra{\beta_r}(I\otimes R(\theta))\ket{\beta_r}\right|^2
=\frac14|\Tr R(\theta)|^2
=\cos^2\theta.
\label{eq:overlap}
\end{equation}

In the ideal valid-event model, Alice applies the binary instrument
\begin{equation}
\Pi_r=\proj{\beta_r},
\qquad
\Pi_r^\perp=I_4-\Pi_r.
\label{eq:binary}
\end{equation}
We denote the reference outcome by $N$ and the complement outcome by $C$. A physical photonic receiver has the larger record alphabet
\begin{equation}
\mathcal Y_{\rm phys}=\{N,C,\varnothing\},
\label{eq:physical-alphabet}
\end{equation}
where $\varnothing$ denotes loss or unverified two-photon occupancy. Such a trial is discarded and restarted. It is never interpreted as either binary projector. Unless stated otherwise, the probabilities below are conditioned on a valid $N/C$ classification.

\begin{figure}[H]
\centering
\resizebox{0.98\linewidth}{!}{%
\begin{tikzpicture}[
>=Latex,
node distance=1.0cm and 0.95cm,
box/.style={draw,rounded corners,minimum width=2.25cm,minimum height=0.82cm,align=center,font=\small},
user/.style={box,fill=gray!10},
alice/.style={box,fill=blue!8},
q/.style={->,thick},
c/.style={->,densely dashed}
]
\node[alice] (src) {Alice: source\\private $r_j$\\fresh $\ket{\beta_{r_j}}$};
\node[user,right=1.0cm of src] (b1) {$B_1$\\$R(s_1\alpha_j)$};
\node[user,right=0.75cm of b1] (b2) {$B_2$\\$R(s_2\alpha_j)$};
\node[right=0.55cm of b2,font=\large] (dots) {$\cdots$};
\node[user,right=0.55cm of dots] (bn) {$B_n$\\$R(s_n\alpha_j)$};
\node[alice,right=1.0cm of bn] (meas) {Alice: readout\\$\{\Pi_{r_j},\Pi_{r_j}^\perp\}$};
\draw[q] (src)--(b1);
\draw[q] (b1)--(b2);
\draw[q] (b2)--(dots);
\draw[q] (dots)--(bn);
\draw[q] (bn)--(meas);
\draw[thick] ($(src.south)+(0,-0.12)$) -- ++(0,-0.9) -| node[pos=0.57,below,font=\scriptsize] {retained subsystem} ($(meas.south)+(0,-0.12)$);
\draw[c] ($(src.north)+(0,0.12)$) -- ++(0,0.68) -| node[pos=0.57,above,font=\scriptsize] {private Bell label $r_j$} ($(meas.north)+(0,0.12)$);
\end{tikzpicture}%
}
\caption{One Bell trial of the $n$-user Loop-Back route. The rotation magnitude $\alpha_j$ may change from trial to trial, but all remote users use the same axis and the same magnitude in a given trial.}
\label{fig:architecture}
\end{figure}
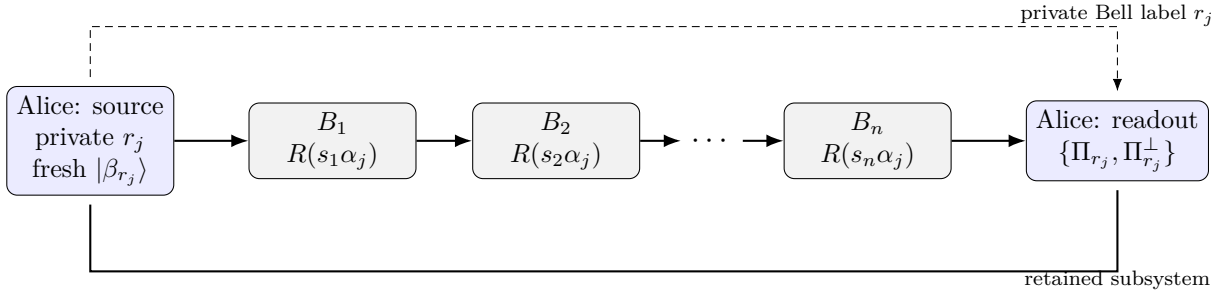

\section{Why a Bell traveler? Algebraic equivalence and information geometry}
\label{sec:bell-boundary}

The ordered-angle theorem is generated by same-axis serial composition, not by entanglement alone. This distinction identifies exactly what the Bell resource contributes and prevents the algebraic result from being misrepresented as an entanglement-only effect.

\subsection{Same-axis composition is carrier-independent}

Let one traveling qubit pass through $B_1,\ldots,B_n$, with user $B_i$ applying $U_i$. The accumulated operation is
\begin{equation}
U_{\rm eff}=U_nU_{n-1}\cdots U_1.
\label{eq:generic-serial}
\end{equation}
For $U_i=R(\theta_i)=e^{-i\theta_iY}$,
\begin{equation}
U_{\rm eff}=R(\Theta),\qquad \Theta=\sum_{i=1}^n\theta_i.
\label{eq:generic-angle-sum}
\end{equation}
With $\theta_i=s_i\alpha$ and $w$ negative signs, $\Theta=(n-2w)\alpha$. The reduction from an $n$-bit pattern to its Hamming weight therefore occurs before any Bell-specific measurement is chosen.

For the Bell reference of equation~\eqref{eq:bellref},
\begin{equation}
P_N^{\rm Bell}(\Theta)=\cos^2\Theta,
\qquad
P_C^{\rm Bell}(\Theta)=\sin^2\Theta.
\label{eq:bell-rot-general}
\end{equation}
Now suppose Alice instead prepares a pure qubit $\ket\psi$ and later tests the returned state against that recorded reference. Writing $r_y=\bra\psi Y\ket\psi$, one obtains
\begin{align}
P_N^{\rm 1q}(\Theta)
&=\left|\bra\psi R(\Theta)\ket\psi\right|^2 \\
&=\cos^2\Theta+r_y^2\sin^2\Theta,
\label{eq:pure-rot-overlap}\\
P_C^{\rm 1q}(\Theta)
&=(1-r_y^2)\sin^2\Theta.
\end{align}
Hence every pure reference satisfying $r_y=0$ reproduces the Bell overlap law exactly. This includes $\ket0$, $\ket1$, $\ket+$, $\ket-$, and real-amplitude linear-polarization states on the $x$--$z$ great circle. The probability-level codeword algebra therefore has a single-qubit realization. The Bell resource is not required merely to obtain equation~\eqref{eq:generic-angle-sum} or the $\cos^2\Theta$ law.

\subsection{What Bell changes}

The physical distinction is local information geometry. For every realized Bell reference,
\begin{equation}
\rho_B=\Tr_A\proj{\beta_r}=\frac{I}{2},
\end{equation}
and after any unitary $U$ on the traveler,
\begin{equation}
U\rho_BU^\dagger=\frac{I}{2}.
\label{eq:bell-local-invariance}
\end{equation}
Thus a party that sees only the traveling subsystem at one location receives no local Bloch direction from the honest Bell carrier. Distinct collective rotations are observable only when the traveler is reunited with Alice's retained subsystem and compared with the private global reference.

A pure-state implementation is different. Even when Alice randomizes a private pure reference so that an observer-averaged ensemble equals $I/2$, each realized conditional state remains rank one. Access to the same carrier at more than one route segment can therefore reveal relative changes that are absent from a single-point density-matrix description. Bell local maximal mixing is consequently a stronger resource-level property than ensemble masking, although it still does not prove secrecy against a coherent probe that couples to the traveler and stores an ancilla across route segments.

For that reason, the protocol and exact-word theorems below remain formulated with Bell references. The pure-qubit calculation is retained only to separate the universal rotation algebra from the Bell-specific information geometry and to clarify why the pure-state construction is a useful precursor rather than the final security model.

\section{Single-pulse \texorpdfstring{$n$}{n}-user family}

Each user holds a private input sign
\begin{equation}
s_i\in\{-1,+1\}
\end{equation}
and applies
\begin{equation}
U_i=R(s_i\alpha).
\end{equation}
Let $w$ be the number of negative signs. Then
\begin{equation}
\sum_{i=1}^n s_i=n-2w
\end{equation}
and
\begin{equation}
U_{\rm eff}(w)=R[(n-2w)\alpha].
\label{eq:ueffw}
\end{equation}

\begin{proposition}[Single-pulse relation law]
For a private Bell reference and the binary instrument in equation~\eqref{eq:binary},
\begin{equation}
P(C\mid w,\alpha)=\sin^2[(n-2w)\alpha].
\label{eq:PCw}
\end{equation}
\end{proposition}
\begin{proof}
Apply equation~\eqref{eq:overlap} to the effective rotation in equation~\eqref{eq:ueffw}. The reference probability is $\cos^2[(n-2w)\alpha]$, the complement probability is its complement.
\end{proof}

The smallest positive angle that makes both unanimous branches deterministic and channel-identical is
\begin{equation}
\boxed{\alpha_n=\frac{\pi}{2n}}.
\label{eq:alphan}
\end{equation}
Then
\begin{equation}
U_{\rm eff}(0)=R(\pi/2)=-iY,
\qquad
U_{\rm eff}(n)=R(-\pi/2)=+iY,
\label{eq:unanchannels}
\end{equation}
so the two unanimous input values induce the same channel up to global phase. Moreover,
\begin{equation}
P(C\mid w=0,n)=1,
\label{eq:complete}
\end{equation}
whereas for a mixed weight
\begin{equation}
\boxed{
P(C\mid w)=\cos^2\left(\frac{\pi w}{n}\right),
\qquad 1\le w\le n-1.
}
\label{eq:singlemimic}
\end{equation}

For $n=2$, equation~\eqref{eq:alphan} gives $\alpha_2=\pi/4$ and an ideal deterministic equality test. This does not contradict the $\alpha=\pi/8$ operating point emphasized in the two-user predecessor~\cite{lizama2026twouser}. That paper preserves the native angle inherited from single-qubit Loop-Back and obtains a probabilistic conclusive event. Here the angle is redesigned to impose a deterministic unanimous signature. The codeword theory then asks how several such signatures separate all mixed $n$-user patterns.

Thus, conditional on a valid binary classification, the ideal single-pulse construction has perfect completeness but nonzero soundness error. The worst mixed case is $w=1$ or $n-1$:
\begin{equation}
P_{\rm max}^{(1)}=\cos^2\left(\frac{\pi}{n}\right)
=1-\frac{\pi^2}{n^2}+O(n^{-4}).
\label{eq:worstsingle}
\end{equation}
Repeating the same pulse $L$ times suppresses that worst case only as
\begin{equation}
\left[\cos^2\left(\frac{\pi}{n}\right)\right]^L
\approx
e^{-\pi^2 L/n^2},
\label{eq:quadraticrep}
\end{equation}
which requires $L=O(n^2\log(1/\varepsilon))$ trials for a fixed worst-case target $\varepsilon$.

\begin{proposition}[All deterministic-unanimity same-axis pulses have integer coefficient]
A same-axis pulse has a deterministic binary Bell result for both unanimous signs if and only if
\begin{equation}
n\alpha=\frac{q\pi}{2}
\end{equation}
for some integer $q$. Equivalently,
\begin{equation}
\alpha=\frac{q\pi}{2n}.
\end{equation}
For every such pulse the two unanimous effective operations are equal up to global phase.
\end{proposition}
\begin{proof}
Under unanimity, equation~\eqref{eq:PCw} gives $P(C)=\sin^2(n\alpha)$. A deterministic binary result requires this probability to be either $0$ or $1$, which occurs exactly when $n\alpha=q\pi/2$ for an integer $q$. For even $q$, $R(\pm q\pi/2)$ are the same multiple of $I$, for odd $q$, they are opposite global phases multiplying $Y$.
\end{proof}

This proposition shows that the integer coefficients used below do not represent an arbitrary discretization. They parameterize the complete set of same-axis pulses that preserve a deterministic unanimous signature under the binary Bell instrument.

\section{Two distinct uses of multiple Bell trials}
\label{sec:two-multipulse}

The central conceptual distinction is shown in Figure~\ref{fig:rep-vs-code}. In both modes, every trial uses a fresh Bell pair and every user keeps the same private input sign over the logical block. What changes is the \emph{angle schedule}.

\begin{definition}[Angle schedule, repetition block, and ordered-angle codeword]
Let
\begin{equation}
\mathbf q=(q_1,\ldots,q_m),\qquad
\alpha_j=\frac{q_j\pi}{2n}.
\label{eq:wordangles}
\end{equation}
The integer sequence $\mathbf q$ is the \emph{angle schedule}. If $q_1=\cdots=q_m=q$, the block is a \emph{same-angle repetition block}. If the coefficients are deliberately selected so that different positions reject different mixed Hamming-weight classes, the schedule is an \emph{ordered-angle codeword}. The expected unanimous binary outcomes form the associated \emph{unanimity signature}.
\end{definition}

The word ``ordered'' has an operational meaning. Trial $j$ has a prescribed angle and therefore a prescribed expected unanimous output. Exchanging two positions does not change the ideal product probability when all trials are eventually performed, but it can change the \emph{early-abort cost}, because a highly discriminating angle may reject a mixed input earlier. Thus the multiset of angles determines ideal full-word discrimination, while their order can matter for resource consumption.

\begin{figure}[H]
\centering
\resizebox{0.98\linewidth}{!}{%
\begin{tikzpicture}[
>=Latex,
box/.style={draw,rounded corners,minimum width=1.45cm,minimum height=0.72cm,align=center,font=\small},
lab/.style={font=\small,align=right},
arr/.style={->,thick}
]
\node[lab] (ra) at (0,1.2) {same-angle\\repetition};
\node[box] (r1) at (2.1,1.2) {$q=1$\\$\alpha$};
\node[box] (r2) at (4.0,1.2) {$q=1$\\$\alpha$};
\node[box] (r3) at (5.9,1.2) {$q=1$\\$\alpha$};
\node[box] (rout) at (8.15,1.2) {$C,C,C$\\same test};
\draw[arr] (r1)--(r2); \draw[arr] (r2)--(r3); \draw[arr] (r3)--(rout);

\node[lab] (ca) at (0,-0.2) {ordered-angle\\codeword};
\node[box] (c1) at (2.1,-0.2) {$q_1=1$\\$\alpha_1$};
\node[box] (c2) at (4.0,-0.2) {$q_2=2$\\$\alpha_2$};
\node[box] (c3) at (5.9,-0.2) {$q_3=4$\\$\alpha_3$};
\node[box] (cout) at (8.15,-0.2) {$C,N,N$\\new queries};
\draw[arr] (c1)--(c2); \draw[arr] (c2)--(c3); \draw[arr] (c3)--(cout);

\node[font=\scriptsize,align=center] at (4.05,-1.15) {same input signs $s_i$ throughout both logical blocks,\\ fresh Bell pair in every trial};
\end{tikzpicture}%
}
\caption{Repetition and coding use several Bell trials differently. Repetition raises the same likelihood to a higher power. An ordered-angle codeword changes the common angle so that different trials interrogate different collective Hamming-weight classes. The $(1,2,4)$ example corresponds to the exact eight-user codeword.}
\label{fig:rep-vs-code}
\end{figure}
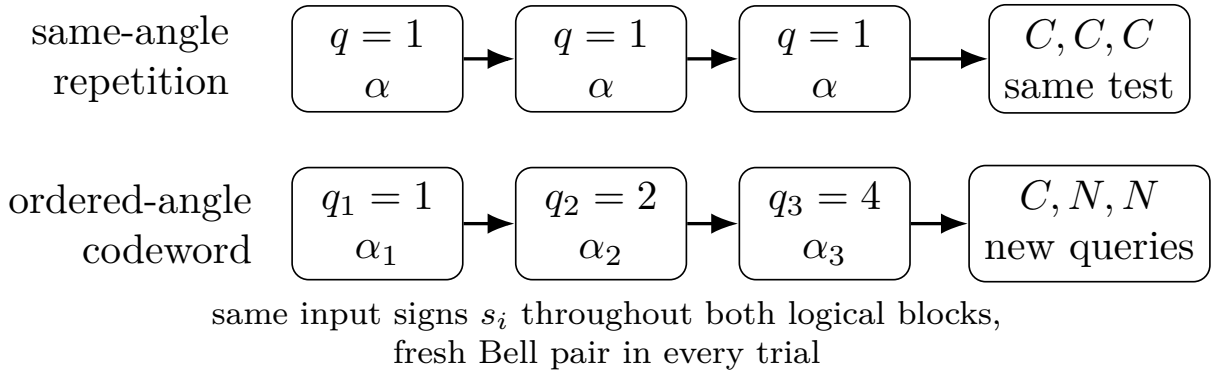

\subsection{Mode A: same-angle repetition}

The simplest block repeats $q=1$:
\begin{equation}
\mathbf q_{\rm rep}^{(L)}=(\underbrace{1,1,\ldots,1}_{L\text{ times}}),
\qquad
\alpha_j=\frac{\pi}{2n}.
\label{eq:repword}
\end{equation}
For a fixed mixed weight $w$, every trial has the same imitation probability
\[
\mu_1(w)=\cos^2\left(\frac{\pi w}{n}\right),
\]
so the full block gives
\begin{equation}
M_{\rm rep}^{(L)}(w)
=\cos^{2L}\left(\frac{\pi w}{n}\right).
\label{eq:rep-mimic}
\end{equation}
Nothing new is measured in the second, third, or later trial. Confidence increases only because independent evidence is accumulated.

For $n=3$, the prime-family construction below contains only $q=1$, so ordinary repetition is already the natural statistical amplifier. For $n\geq4$, angle diversity can ask genuinely different collective questions at the same pulse budget.

\subsection{Mode B: angle-diverse codewords}

Changing the angle can produce strictly stronger discrimination than repeating one setting. The shortest example occurs at $n=4$. Two $q=1$ trials leave a worst mixed imitation probability of $1/4$, while the two-position word $\mathbf q=(1,2)$ uses the signature $(C,N)$ and rejects every mixed Hamming weight in the ideal model. The complete four-user calculation is given in section~\ref{sec:exact-powers}. This fixed-budget contrast is the basic reason to treat angular diversity as code structure rather than as repetition.

\section{General ordered-angle codeword theory}

The definitions in Section~\ref{sec:two-multipulse} now admit a general likelihood analysis. For the angle schedule
\[
\alpha_j=\frac{q_j\pi}{2n},
\]
unanimity produces the deterministic signature
\begin{equation}
y_j^{\rm U}=
\begin{cases}
C,&q_j\text{ odd},\\
N,&q_j\text{ even},
\end{cases}
\label{eq:signature}
\end{equation}
while a mixed Hamming weight $w$ matches that one position with probability
\begin{equation}
\mu_j(w)=\cos^2\left(\frac{\pi q_jw}{n}\right).
\label{eq:mimicpulse}
\end{equation}
The full-word behavior follows by multiplying the independent trial probabilities.

\begin{theorem}[Pulse-word imitation law]
Assume independent Bell trials and condition on valid binary Bell classifications. A mixed configuration of Hamming weight $w$ imitates the entire unanimous signature of word $\mathbf q$ with probability
\begin{equation}
\boxed{
M_{\mathbf q}(w)
=
\prod_{j=1}^{m}
\cos^2\left(\frac{\pi q_jw}{n}\right).
}
\label{eq:wordmimic}
\end{equation}
For true unanimity,
\begin{equation}
M_{\mathbf q}(0)=M_{\mathbf q}(n)=1.
\end{equation}
\end{theorem}

The design problem can therefore be written as
\begin{equation}
\min_{\mathbf q}
\max_{1\le w\le n-1}
M_{\mathbf q}(w),
\label{eq:minimax}
\end{equation}
subject to a chosen pulse budget $m$ and implementation constraints.

\subsection{Aggregate acceptance and accepted-block mismatch}

If the $n$ input bits are independently uniform, Hamming weight $w$ occurs with probability
\[
2^{-n}\binom nw.
\]
A word is accepted only if its full binary record equals the unanimous signature. Define
\begin{equation}
S_{\mathbf q}
=
\sum_{w=1}^{n-1}
\binom nw M_{\mathbf q}(w).
\label{eq:Sq}
\end{equation}
Then
\begin{equation}
\boxed{
P_{\rm acc}(\mathbf q)
=
2^{-n}\left(2+S_{\mathbf q}\right),
}
\label{eq:paccq}
\end{equation}
and the intrinsic fraction of accepted blocks that were actually mixed is
\begin{equation}
\boxed{
\epsilon(\mathbf q)
=
\frac{S_{\mathbf q}}{2+S_{\mathbf q}}.
}
\label{eq:epsq}
\end{equation}

These equations separate two quantities that should not be conflated. $M_{\mathbf q}(w)$ is a relation-discrimination probability for a fixed mixed configuration. $\epsilon(\mathbf q)$ also includes the much larger number of mixed input strings in the prior distribution.

\section{Exact and pulse-minimal family for powers of two}
\label{sec:exact-powers}

The most striking case occurs when the number of users is a power of two.

\begin{theorem}[Existence of a perfect finite word]
For finite nonadaptive schedules drawn from the complete deterministic-unanimity same-axis pulse family characterized above, a word satisfying
\begin{equation}
M_{\mathbf q}(w)=0
\qquad
\text{for every }1\le w\le n-1
\label{eq:perfect}
\end{equation}
exists if and only if
\begin{equation}
\boxed{n=2^r}
\end{equation}
for some integer $r\ge1$.
\end{theorem}

\begin{proof}
First let $n=2^r$. Choose
\begin{equation}
\mathbf q_{2^r}=(1,2,4,\ldots,2^{r-1}).
\label{eq:powerword}
\end{equation}
For any mixed weight $w$, write
\[
w=2^v u,
\]
where $u$ is odd and $0\le v\le r-1$. The word contains
\[
q=2^{r-v-1}.
\]
For that trial,
\[
\frac{qw}{n}
=
\frac{2^{r-v-1}2^vu}{2^r}
=
\frac{u}{2},
\]
so
\[
\cos^2\left(\frac{\pi qw}{n}\right)
=
\cos^2\left(\frac{u\pi}{2}\right)=0.
\]
Therefore every mixed weight fails at least one pulse and the word is perfect.

Conversely, write a non-power-of-two $n$ as
\[
n=2^r d,
\qquad
d>1\ \text{odd}.
\]
Consider the mixed weight
\[
w=2^r<n.
\]
For any integer $q$,
\[
\frac{qw}{n}=\frac{q}{d}.
\]
A zero mimic factor would require
\[
\cos\left(\frac{\pi q}{d}\right)=0,
\]
equivalently $2q=d(2k+1)$. The left side is even and the right side is odd because $d$ is odd, a contradiction. Thus this mixed weight has nonzero imitation probability in every allowed trial, so no finite perfect word exists.
\end{proof}

Thus, within the complete deterministic-unanimity same-axis pulse family and nonadaptive finite schedules, \emph{power-of-two user counts are exactly those for which finite exact discrimination is possible}. This is not an impossibility statement for arbitrary real-angle schedules with nondeterministic positions, adaptive measurements, or higher-dimensional resources.

\begin{theorem}[Pulse minimality for $n=2^r$]
Any perfect word for $n=2^r$ requires at least $r$ trials. The word in equation~\eqref{eq:powerword} is therefore pulse-minimal:
\begin{equation}
\boxed{m_{\min}=r=\log_2 n.}
\end{equation}
\end{theorem}

\begin{proof}
For $n=2^r$, a zero factor in a trial with integer $q$ requires
\[
\frac{qw}{2^r}=k+\frac12,
\]
hence
\[
\nuTwo(q)+\nuTwo(w)=r-1.
\]
A fixed $q$ can therefore eliminate mixed weights from only one $2$-adic valuation class $\nuTwo(w)$. The mixed weights $1,\ldots,2^r-1$ contain all $r$ classes
\[
\nuTwo(w)=0,1,\ldots,r-1.
\]
At least $r$ pulses are necessary. Equation~\eqref{eq:powerword} uses exactly one pulse for each class.
\end{proof}

\subsection{Worked example: four users}

For $n=4$, the minimal exact word is
\begin{equation}
\mathbf q=(1,2),
\qquad
\boldsymbol{\alpha}=
\left(\frac{\pi}{8},\frac{\pi}{4}\right).
\end{equation}
The unanimous signature is
\begin{equation}
\boxed{(C,N)}.
\end{equation}
Table~\ref{tab:n4} shows the complete Hamming-weight logic.

\begin{table}[H]
\centering
\caption{Exact four-user two-pulse word. ``Match'' means that the trial reproduces the outcome expected under unanimity.}
\label{tab:n4}
\small
\begin{tabular}{ccccc}
\toprule
Weight $w$ & Multiplicity & $q=1$: $P(\text{match})$ & $q=2$: $P(\text{match})$ & Full word \\
\midrule
0 or 4 & $2$ & $1$ & $1$ & accept \\
1 or 3 & $8$ & $1/2$ & $0$ & reject with certainty \\
2 & $6$ & $0$ & $1$ & reject with certainty \\
\bottomrule
\end{tabular}
\end{table}

The result is important because a single $\pi/8$ pulse is not exact for four users. Its worst mixed case has probability $1/2$ of imitating unanimity. Repeating $\pi/8$ twice still leaves a worst-case imitation probability $1/4$. The ordered second pulse at $\pi/4$ instead annihilates the remaining odd-weight ambiguity and makes the complete two-pulse test deterministic.

The fixed two-pulse comparison can also be expressed at the accepted-block level. Under independent unbiased input bits, $\mathbf q=(1,1)$ leaves an intrinsic mixed fraction
\[
\epsilon(1,1)=\frac12,
\]
whereas $\mathbf q=(1,2)$ has $\epsilon(1,2)=0$ in the ideal algebraic model. Thus the gain is not a consequence of using two pulses rather than one. It comes from replacing the repeated query by a complementary angular query.

\subsection{Worked example: eight users}

For $n=8$,
\begin{equation}
\mathbf q=(1,2,4),
\qquad
\boldsymbol{\alpha}
=
\left(\frac{\pi}{16},\frac{\pi}{8},\frac{\pi}{4}\right),
\end{equation}
with unanimous signature
\begin{equation}
\boxed{(C,N,N)}.
\end{equation}
The rejection mechanism is most transparent through the $2$-adic valuation:
\begin{itemize}[leftmargin=2em]
\item If $w$ is odd, $q=4$ gives $qw/8=w/2$ and hence a zero mimic factor.
\item If $w=2u$ with $u$ odd, $q=2$ gives $qw/8=u/2$ and a zero factor.
\item The remaining mixed class is $w=4$, and $q=1$ gives $qw/8=1/2$.
\end{itemize}
Every mixed Hamming weight is therefore rejected with certainty in only three Bell trials.

\begin{table}[H]
\centering
\caption{Power-of-two exact words and the pulse budget required by simple $q=1$ repetition to make the \emph{worst-case} mixed imitation smaller than $1\%$.}
\label{tab:pulsegain}
\small
\begin{tabular}{ccccc}
\toprule
$n$ & Exact word $\mathbf q$ & Exact trials & Repeated $q=1$ trials for $<1\%$ & Reduction \\
\midrule
4 & $(1,2)$ & 2 & 7 & $3.5\times$ \\
8 & $(1,2,4)$ & 3 & 30 & $10\times$ \\
16 & $(1,2,4,8)$ & 4 & 119 & $29.75\times$ \\
\bottomrule
\end{tabular}
\end{table}

\begin{figure}[H]
\centering
\includegraphics[width=0.78\linewidth]{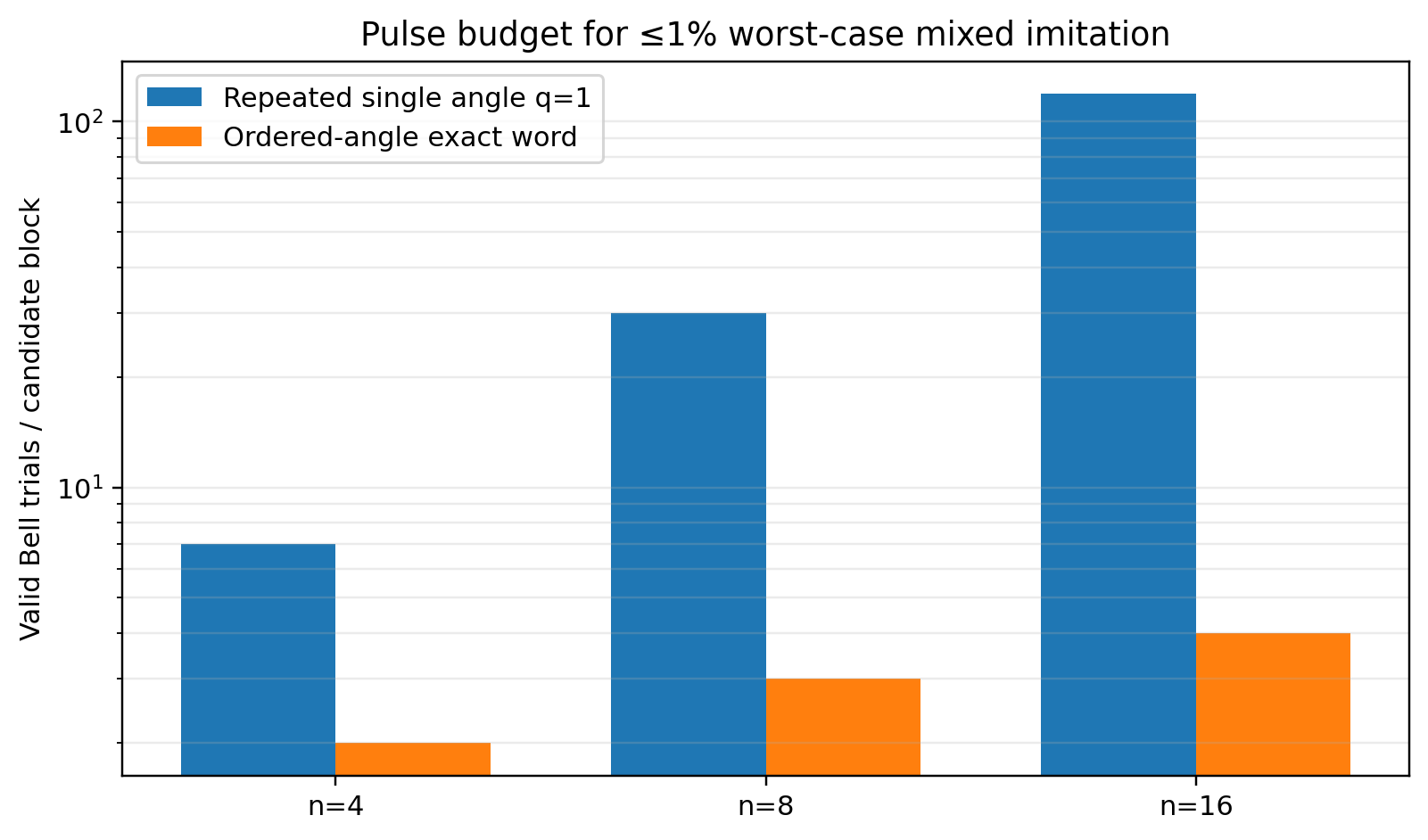}
\caption{For power-of-two user counts, ordered-angle words replace the quadratic-in-$n$ repetition burden by an exact $\log_2n$ pulse word. The comparison uses the same worst-case mixed-imitation target of $1\%$ for the repeated single-angle baseline.}
\label{fig:pulsebudget}
\end{figure}

\section{Calibration and binary-readout robustness of exact words}
\label{sec:robustness}

The exactness theorem is algebraic. A mixed Hamming weight is rejected because at least one ideal word position has zero probability of matching the unanimous signature. A physical implementation must quantify how that zero is lifted by imperfect rotations and imperfect binary classification. The following model isolates those two effects without claiming a complete optical-noise or security treatment.

\subsection{Additive per-user angular offsets}

Let the intended operation of user $i$ in trial $j$ be $R(s_i\alpha_j)$. We write the actual same-axis operation as
\begin{equation}
\widetilde U_i^{(j)}=R(s_i\alpha_j+\delta_{ij}),
\label{eq:actual-rotation}
\end{equation}
where $\delta_{ij}$ is an additive amplitude-space angle error. This model includes a sign-correlated magnitude error as the special case $\delta_{ij}=s_i\delta_{ij}^{\rm mag}$. Same-axis composition remains exact and gives
\begin{equation}
\widetilde U_{\rm eff}^{(j)}
=R\!\left[(n-2w)\alpha_j+\Delta_j\right],
\qquad
\Delta_j=\sum_{i=1}^{n}\delta_{ij}.
\label{eq:delta-sum}
\end{equation}
For Hamming weight $w$, the probability that trial $j$ reproduces the \emph{ideal} unanimous outcome is
\begin{equation}
\widetilde\mu_j(w)
=
\cos^2\!\left(\frac{\pi q_jw}{n}-\Delta_j\right).
\label{eq:perturbed-match}
\end{equation}
The sign of $\Delta_j$ is immaterial for the worst-case bounds below.

If an exact word rejects weight $w$ at position $j_*$, then $\pi q_{j_*}w/n$ is an odd multiple of $\pi/2$. Equation~\eqref{eq:perturbed-match} therefore reduces to
\begin{equation}
\boxed{\widetilde\mu_{j_*}(w)=\sin^2\Delta_{j_*}.}
\label{eq:lifted-zero}
\end{equation}
Thus the exact zero is lifted quadratically for small accumulated angular error. If
\begin{equation}
|\delta_{ij}|\leq\delta_{\max},
\end{equation}
then
\begin{equation}
|\Delta_j|\leq n\delta_{\max}
\label{eq:delta-bound}
\end{equation}
and, for $n\delta_{\max}\leq\pi/2$,
\begin{equation}
\widetilde\mu_{j_*}(w)
\leq \sin^2(n\delta_{\max})
= n^2\delta_{\max}^2+O(n^4\delta_{\max}^4).
\label{eq:angular-zero-bound}
\end{equation}
For a unanimous input, every ideal position matches with probability one and the perturbed match probability is $\cos^2\Delta_j$. Hence an $m$-position word obeys the conservative ideal-readout bound
\begin{equation}
P_{\rm true}\geq \cos^{2m}(n\delta_{\max}).
\label{eq:true-angular}
\end{equation}

\subsection{Symmetric binary classification error and contrast}

Condition on a valid two-photon event and let $\eta_j$ be the probability that the ideal $N/C$ label is flipped at trial $j$. If the pre-readout match probability is $\widetilde\mu_j$, the recorded match probability is
\begin{equation}
\overline\mu_j
=(1-\eta_j)\widetilde\mu_j+\eta_j(1-\widetilde\mu_j)
=\eta_j+(1-2\eta_j)\widetilde\mu_j.
\label{eq:readout-flip}
\end{equation}
Equivalently, with binary contrast $V_j=1-2\eta_j$,
\begin{equation}
\overline\mu_j=\frac12\left[1+V_j(2\widetilde\mu_j-1)\right].
\label{eq:contrast-form}
\end{equation}
This form connects the abstract flip model to the contrast of the inherited binary Bell comparator without assuming a particular detector implementation.

\begin{proposition}[Robustness bounds for an exact word]
Assume an ideal exact word, independent valid Bell trials, $|\delta_{ij}|\leq\delta_{\max}$, and $0\leq\eta_j\leq\eta_{\max}<1/2$. If $n\delta_{\max}\leq\pi/4$, then every mixed input satisfies
\begin{equation}
\boxed{
P_{\rm FA}
\leq
\eta_{\max}+(1-2\eta_{\max})\sin^2(n\delta_{\max}).
}
\label{eq:false-accept-bound}
\end{equation}
For a unanimous input, an $m$-position word satisfies
\begin{equation}
\boxed{
P_{\rm true}
\geq
\left[
\eta_{\max}+(1-2\eta_{\max})\cos^2(n\delta_{\max})
\right]^m.
}
\label{eq:true-accept-bound}
\end{equation}
\end{proposition}
\begin{proof}
Every mixed input to an exact word has at least one ideal rejecting position. Equations~\eqref{eq:lifted-zero} and~\eqref{eq:delta-bound} bound the pre-readout match at that position by $\sin^2(n\delta_{\max})$. Equation~\eqref{eq:readout-flip} is increasing in both $\widetilde\mu_j$ and $\eta_j$ when $\widetilde\mu_j\leq1/2$, which follows from $n\delta_{\max}\leq\pi/4$. All other position-match probabilities are at most one, proving equation~\eqref{eq:false-accept-bound}. Under unanimity, $\widetilde\mu_j=\cos^2\Delta_j\geq\cos^2(n\delta_{\max})\geq1/2$. In this range equation~\eqref{eq:readout-flip} decreases with $\eta_j$, so the worst case is $\eta_j=\eta_{\max}$. Multiplying the $m$ independent valid-trial bounds proves equation~\eqref{eq:true-accept-bound}.
\end{proof}

For small $\eta_{\max}$ and $\delta_{\max}$, the leading behavior is
\begin{align}
P_{\rm FA}&\lesssim \eta_{\max}+n^2\delta_{\max}^2,\\
P_{\rm true}&\gtrsim 1-m\eta_{\max}-mn^2\delta_{\max}^2.
\end{align}
The bounds are intentionally conservative because they allow all user offsets in one trial to add coherently.

\subsection{Numerical tolerance scale}

With ideal binary readout, imposing the conservative target $\sin^2(n\delta_{\max})<1\%$ gives
\begin{equation}
\delta_{\max}<\frac{\arcsin(0.1)}{n}.
\label{eq:delta-1pct}
\end{equation}
Table~\ref{tab:calibration} lists the corresponding scale. The conventional Bloch-sphere rotation angle is $2\theta$, so its tolerance is twice the amplitude-space value shown in the third column.

\begin{table}[H]
\centering
\caption{Per-user calibration bound from the conservative condition $\sin^2(n\delta_{\max})<1\%$ with ideal binary readout. The final column gives the conventional Bloch-sphere angle error $2\delta_{\max}$.}
\label{tab:calibration}
\small
\begin{tabular}{cccc}
\toprule
Users $n$ & exact word length $m$ & $\delta_{\max}$ in $\theta$ & Bloch-angle error $2\delta_{\max}$ \\
\midrule
4  & 2 & $1.435^\circ$ & $2.870^\circ$ \\
8  & 3 & $0.717^\circ$ & $1.435^\circ$ \\
16 & 4 & $0.359^\circ$ & $0.717^\circ$ \\
\bottomrule
\end{tabular}
\end{table}

\begin{figure}[H]
\centering
\includegraphics[width=0.78\linewidth]{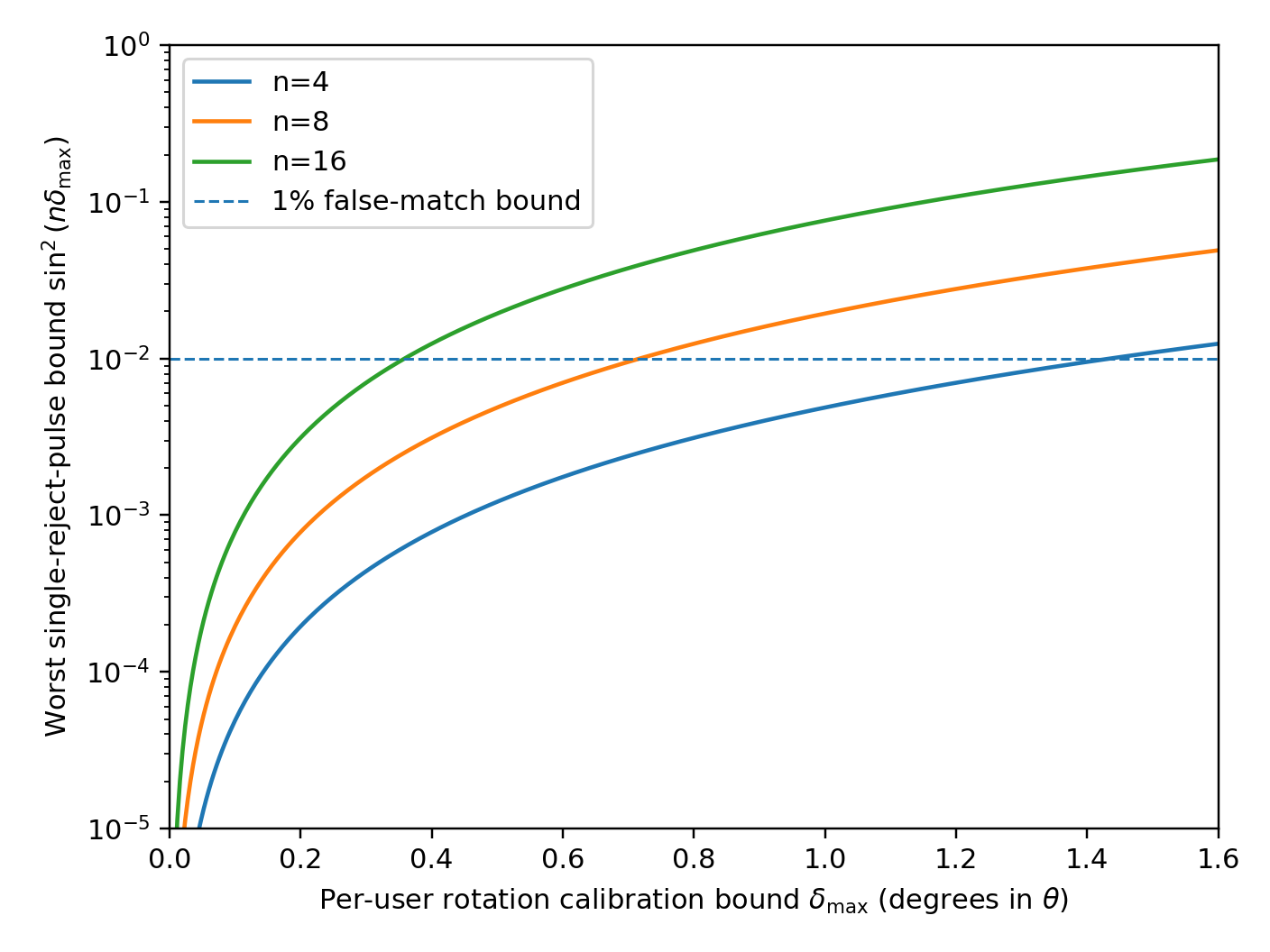}
\caption{Conservative residual-match bound $\sin^2(n\delta_{\max})$ for an ideal rejecting position under bounded additive per-user angular offsets. The curves show the $1/n$ tightening of calibration tolerance at fixed discrimination target.}
\label{fig:calibration-robustness}
\end{figure}

\section{Balanced statistical words for odd prime user counts}

If $n$ is not a power of two, the previous theorem rules out a perfect finite word within the deterministic-unanimity same-axis pulse family. Odd primes nevertheless admit a particularly symmetric statistical construction.

Let
\begin{equation}
n=p,
\qquad
p\ \text{odd prime},
\end{equation}
and choose the word
\begin{equation}
\boxed{
\mathbf q_p=
\left(1,2,\ldots,\frac{p-1}{2}\right).
}
\label{eq:primeword}
\end{equation}

\begin{theorem}[Balanced odd-prime word]\label{thm:primeword}
For every mixed weight $w=1,\ldots,p-1$,
\begin{equation}
\boxed{
M_{\mathbf q_p}(w)=2^{-(p-1)}.
}
\label{eq:primeuniform}
\end{equation}
Thus every nonunanimous Hamming weight has exactly the same probability of imitating the full unanimous signature.
\end{theorem}

\begin{proof}
Because $p$ is prime, multiplication by any nonzero $w$ permutes the nonzero residues modulo $p$. Since $\cos^2(\pi x/p)=\cos^2[\pi(p-x)/p]$, the set of squared cosine factors obtained from
\[
q=1,\ldots,\frac{p-1}{2}
\]
is independent of $w$. Therefore
\[
M_{\mathbf q_p}(w)
=
\prod_{q=1}^{(p-1)/2}
\cos^2\left(\frac{\pi q}{p}\right).
\]
The standard product identity
\[
\prod_{q=1}^{(p-1)/2}
\cos\left(\frac{\pi q}{p}\right)
=
2^{-(p-1)/2}
\]
then yields equation~\eqref{eq:primeuniform}.
\end{proof}

\subsection{Repeated prime words}

Repeat the complete prime word $r$ times with fresh Bell pairs. A mixed configuration then imitates all $r$ words with probability
\begin{equation}
M_{p,r}=2^{-r(p-1)}.
\label{eq:primeMr}
\end{equation}
Since there are $2^p-2$ mixed input strings and two unanimous strings,
\begin{equation}
\boxed{
\epsilon_{p,r}
=
\frac{(2^p-2)2^{-r(p-1)}}
{2+(2^p-2)2^{-r(p-1)}}.
}
\label{eq:primeeps}
\end{equation}
The total number of Bell trials is
\begin{equation}
m_{p,r}=r\frac{p-1}{2}.
\label{eq:primepulses}
\end{equation}

\begin{corollary}[The three-user case]
For $p=3$, the prime word contains only $q=1$, so one prime word is one $\alpha=\pi/6$ Bell trial. Equation~\eqref{eq:primeeps} reduces to $\epsilon_{3,r}=3/(4^r+3)$. Thus the three-user construction is ordinary same-angle repetition, while larger primes admit angle-diverse balanced words.
\end{corollary}

\subsection{Worked example: five users}

For $p=5$,
\begin{equation}
\mathbf q_5=(1,2),
\qquad
\boldsymbol{\alpha}
=
\left(\frac{\pi}{10},\frac{\pi}{5}\right),
\end{equation}
and the unanimous signature is again $(C,N)$. For $w=1$,
\[
M_{\mathbf q_5}(1)
=
\cos^2\left(\frac{\pi}{5}\right)
\cos^2\left(\frac{2\pi}{5}\right)
=
\frac1{16}.
\]
For $w=2$, the two factors exchange values, so the product remains $1/16$. By symmetry the same holds for $w=3,4$.

The importance of changing the second angle is visible at a fixed physical pulse budget. Table~\ref{tab:n5budget} compares same-angle repetition with the balanced codeword.

\begin{table}[H]
\centering
\caption{Five-user comparison at equal Bell-trial budgets. The repetition rows use only $q=1$, the codeword rows repeat the complete designed word $(1,2)$.}
\label{tab:n5budget}
\small
\begin{tabularx}{\linewidth}{P{2.0cm} P{2.3cm} P{2.25cm} P{2.4cm} Y}
\toprule
Trials & Angle word & Unanimous signature & Worst fixed mixed mimic & Mixed fraction among accepted blocks \\
\midrule
2 & $(1,1)$ & $(C,C)$ & $0.42838$ & $0.69070$ \\
2 & $(1,2)$ & $(C,N)$ & $1/16=0.0625$ & $15/31\approx0.48387$ \\
4 & $(1,1,1,1)$ & $(C,C,C,C)$ & $0.18351$ & $0.47873$ \\
4 & $(1,2,1,2)$ & $(C,N,C,N)$ & $1/256\approx0.00391$ & $15/271\approx0.05535$ \\
\bottomrule
\end{tabularx}
\end{table}

Two lessons follow. First, the two-pulse codeword $(1,2)$ is substantially stronger than two repetitions $(1,1)$ even though both consume two Bell pairs. Second, four identical pulses still leave nearly half of accepted blocks mixed under the independent-bit prior, whereas repeating the two-angle codeword once more reduces that posterior contamination below $6\%$. The gain comes from the diversity of the angular questions and is then amplified by repetition of the \emph{whole codeword}.

The table also shows why fixed-pattern imitation and aggregate posterior error must be reported separately. For five users, one $(1,2)$ word gives $\epsilon_{5,1}=15/31\approx0.4839$, while two complete words give $\epsilon_{5,2}=15/271\approx0.05535$. A third complete word, using six Bell trials, gives $\epsilon_{5,3}\approx0.00365$. These values follow directly from equation~\eqref{eq:primeeps}.

The general prime-word scaling is summarized in Figure~\ref{fig:primeamp} and Table~\ref{tab:prime}. They show why the complete codeword, rather than an individual pulse, is the natural unit of statistical amplification.

\begin{figure}[H]
\centering
\includegraphics[width=0.78\linewidth]{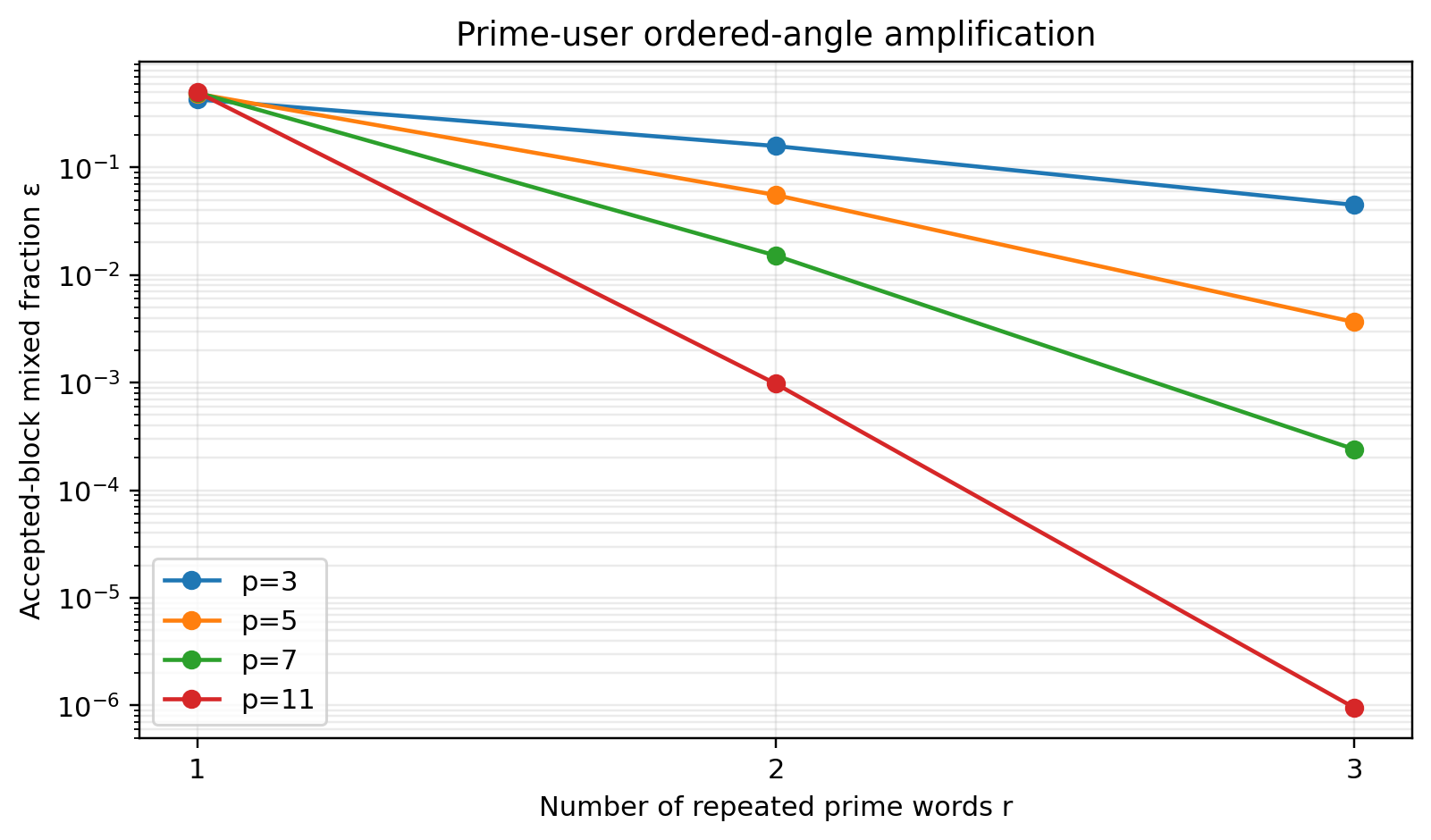}
\caption{Accepted-block mixed fraction for the balanced odd-prime word. Larger primes have more mixed input strings, so one word alone leaves an aggregate error near one half. Repeating the complete word rapidly suppresses the contamination.}
\label{fig:primeamp}
\end{figure}

\begin{table}[H]
\centering
\caption{Balanced prime-word examples. The final column is the aggregate mixed fraction among accepted blocks under independent unbiased input bits.}
\label{tab:prime}
\small
\begin{tabular}{cccccc}
\toprule
$p$ & word length & repetitions $r$ & total trials & fixed mixed mimic & $\epsilon_{p,r}$ \\
\midrule
3 & 1 & 3 & 3 & $1/64$ & $0.04478$ \\
5 & 2 & 2 & 4 & $1/256$ & $0.05535$ \\
5 & 2 & 3 & 6 & $1/4096$ & $0.00365$ \\
7 & 3 & 2 & 6 & $1/4096$ & $0.01515$ \\
11 & 5 & 2 & 10 & $1/2^{20}$ & $0.000975$ \\
\bottomrule
\end{tabular}
\end{table}

\section{Composite user counts with odd factors}
\label{sec:composite}

The impossibility direction of the exactness theorem identifies the obstruction for every $n$ with an odd factor. Write $n=2^r d$ with odd $d>1$. The mixed weight $w=2^r$ gives $qw/n=q/d$ for every deterministic-unanimity coefficient $q$. An odd denominator cannot make this ratio half-integer, so no allowed position can annihilate that weight exactly. The obstruction is therefore arithmetic rather than a failure to search a sufficiently long word.

Useful finite-budget words still exist. For example, at $n=6$ the short schedule $\mathbf q=(1,2,3)$ eliminates several Hamming classes exactly and suppresses the unavoidable classes statistically. For composite $n$ with odd factors, the natural design task is therefore the finite-budget minimax problem in equation~\eqref{eq:minimax}, or its prior-weighted analogue based on $\epsilon(\mathbf q)$, rather than a search for a nonexistent finite perfect word.

\section{Resource accounting and the meaning of gain}

Multiple pulses improve the \emph{quality of the relation test}. Angle diversity can improve it much more rapidly than equal-angle repetition. Neither mechanism creates additional prior probability that independently chosen user bits are all equal. Under independent unbiased choices,
\begin{equation}
\boxed{
P(\mathrm{unanimity})=\frac{2}{2^n}=2^{1-n}.
}
\label{eq:priorU}
\end{equation}
This combinatorial factor is untouched by the pulse word.

Consequently, three distinct rates should be separated:
\begin{enumerate}[leftmargin=2em]
\item \textbf{relation discrimination:} how small $M_{\mathbf q}(w)$ or $\epsilon(\mathbf q)$ becomes,
\item \textbf{accepted logical blocks per input:} $P_{\rm acc}(\mathbf q)$,
\item \textbf{true common bits per physical Bell trial:} the prior true-unanimity contribution divided by the number of physical trials, before optical loss and security processing.
\end{enumerate}

For a fixed $m$-pulse word with exact discrimination,
\begin{equation}
Y_{\rm U,pulse}^{\rm alg}
=
\frac{2^{1-n}}{m}.
\label{eq:trueyield}
\end{equation}
This is an algebraic utilization factor, not a secret-key rate.

For a fixed-length implementation, let $p_{{\rm valid},j}$ be the honest probability that trial $j$ produces a verified two-photon $N/C$ record. The true accepted-block probability of an exact word becomes
\begin{equation}
P_{\rm true,acc}^{\rm phys}
=2^{1-n}\prod_{j=1}^{m}p_{{\rm valid},j}.
\label{eq:physical-true-accept}
\end{equation}
This factorization assumes setting-independent validity events. Adversarial or setting-dependent loss requires a separate analysis.

\subsection{Equal-budget interpretation}

The verification gain is already visible in Table~\ref{tab:pulsegain} and Figure~\ref{fig:pulsebudget}. At the same worst-case discrimination target, the exact dyadic word uses three trials rather than 30 for $n=8$, and four rather than 119 for $n=16$. This is a gain in relation verification, not in the prior probability that independently chosen inputs are unanimous.

\subsection{Sequential early abort}

Because the unanimous signature is known in advance, an input block can be rejected as soon as one position fails to match. This can reduce the mean number of valid Bell trials spent on mixed inputs. It does not change the full-word discrimination law and it is not privacy-neutral, because the stopping position can reveal additional Hamming-weight-class information to Alice.

For this reason, early abort is treated only as an optional resource optimization. The full enumeration of expected valid-trial and physical-attempt costs is given in Appendix~\ref{sec:earlyabort-app}. Applications that prioritize transcript uniformity should execute the full word and disclose only the final pass/fail decision, while recognizing that Alice still holds the internal per-position record.

\section{Protocol specification}
\label{sec:protocol}

One logical execution proceeds as follows.
\begin{enumerate}[leftmargin=2em]
\item \textbf{Input selection.} Each user $B_i$ samples or receives a private input bit $b_i$ and maps it to $s_i=(-1)^{b_i}$. The same sign is retained across the complete word.
\item \textbf{Public word.} Alice and the users agree on $\mathbf q=(q_1,\ldots,q_m)$ and $\alpha_j=q_j\pi/(2n)$.
\item \textbf{Fresh Bell trial.} At position $j$, Alice samples a private Bell label $r_j$, prepares $\ket{\beta_{r_j}}$, retains one subsystem, and sends the traveler through $B_1,\ldots,B_n$.
\item \textbf{Serial transformation.} User $B_i$ applies $R(s_i\alpha_j)$ and forwards the traveler. The same transformation hardware is used at every word position, with only the public magnitude changed.
\item \textbf{Binary readout.} Alice compares the returned Bell pair with her private reference and records $N$, $C$, or the erasure symbol $\varnothing$. An erasure is never reinterpreted as a binary outcome.
\item \textbf{Relation decision.} A valid fixed-length word passes only if its complete $N/C$ record equals the unanimous signature in equation~\eqref{eq:signature}. Early rejection is optional. It saves trials but reveals the first failing position to Alice.
\item \textbf{Interpretation.} For an exact power-of-two word in the ideal model, pass is equivalent to $b_1=\cdots=b_n$. For a statistical word, pass has the residual mixed fraction $\epsilon(\mathbf q)$ and must not be treated as exact unanimity without an additional error criterion.
\item \textbf{Optional raw-key use.} Only after an exact relation pass, or after a separately specified statistical confidence test, may each user retain its own $b_i$ as a raw common-key candidate. This optional interpretation does not constitute a conference-key security proof.
\end{enumerate}

\section{Operational interpretation and optional conference-key use}
\label{sec:agreement}

The ordered word is first a relation primitive. Let $\mathbf b=(b_1,\ldots,b_n)\in\{0,1\}^n$ and define
\begin{equation}
u(\mathbf b)=
\begin{cases}
1,&b_1=b_2=\cdots=b_n,\\
0,&\text{otherwise}.
\end{cases}
\label{eq:unanimity-predicate}
\end{equation}
For an exact power-of-two word, ideal acceptance is equivalent to $u(\mathbf b)=1$. For statistical words the same decision has the residual mixed fraction quantified by equation~\eqref{eq:epsq}. This relation layer can then support two different application interpretations without changing the quantum hardware.

As an optional \emph{conference-key-candidate} use, each user retains its own local $b_i$ after an exact word passes. Every honest participant then knows that all local values coincide and can retain that bit as one raw common candidate. This is not yet a CKA security theorem: parameter estimation, authenticated conference reconciliation, privacy amplification, malicious-user fairness, and finite-size analysis remain additional layers.

In \emph{value-private unanimity mode}, Alice announces only whether the complete signature passed. User $B_i$ combines that public relation result with its own bit and, upon exact acceptance, infers
\begin{equation}
b_1=b_2=\cdots=b_n=b_i.
\label{eq:user-inference}
\end{equation}
Alice does not learn the common value from the final modeled channel because unanimous $0^n$ and $1^n$ are channel-identical, as shown below. The claim is deliberately narrower than general quantum private comparison: rejected outcome words may still reveal aggregate Hamming-weight information.

For four users, the exact word $\mathbf q=(1,2)$ has unanimous signature $(C,N)$. The two inputs $0000$ and $1111$ therefore produce the same pass result at Alice, whereas every $3+1$ or $2+2$ split is rejected in the ideal model. Each participant nevertheless knows which unanimous branch occurred because it knows its own input. This example is sufficient to illustrate the functionality. The primitive is not a voting, majority-consensus, or Byzantine-agreement protocol.

\section{Privacy properties and security boundaries}

\subsection{Common-sign final-channel identity}

The ordered-word construction has an exact common-sign privacy property that extends across the full codeword family.

\begin{proposition}[Final-station indistinguishability of unanimous bits]
For every integer word $\mathbf q$, let $T_1,\ldots,T_m$ denote the single-qubit travelers used in its $m$ trials, each confined to the admitted optical mode, and let $R$ be an arbitrary ancilla retained by Alice. The complete channel induced on an arbitrary joint state $\rho_{R T_1\cdots T_m}$ by unanimous $0^n$ and unanimous $1^n$ input blocks is identical, provided Alice has no access to a traveler between the user operations of its trial.
\end{proposition}

\begin{proof}
In trial $j$, unanimity gives
\[
U_{0^n}^{(j)}=R(q_j\pi/2),
\qquad
U_{1^n}^{(j)}=R(-q_j\pi/2).
\]
If $q_j$ is even, both are the same multiple of $I$. If $q_j$ is odd, they are opposite global phases multiplying $Y$. Thus
\[
U_{0^n}^{(j)}\doteqph U_{1^n}^{(j)}
\]
for every trial. Their tensor-product channels over the complete word are therefore identical.
\end{proof}

This is stronger than assuming that Alice follows the coarse-grained comparator. No final measurement on the modeled returned qubits and retained ancilla can distinguish the two unanimous values. The scope is nevertheless exact: the result does not cover multiphoton, higher-dimensional, out-of-mode, Trojan-horse, or side-channel probes, and it does not cover access between user operations.

\subsection{Transcript information for nonunanimous inputs}
\label{sec:transcript}

The proposition protects the value of an accepted unanimous block. It is not a full input-privacy theorem. For a mixed weight $w$, the position probabilities $\mu_j(w)$ depend on $w$. The complete $N/C$ record can therefore distinguish some Hamming-weight classes statistically, and an early-abort implementation additionally reveals the first position at which the unanimous signature failed. Neither record identifies which users supplied the negative signs because the same-axis composition depends only on $w$, but it may reveal more than the single predicate bit $u=0$.

If an application requires the mediator to learn only pass/fail for both accepted and rejected blocks, the present receiver and early-abort rule are insufficient. A fixed-length execution removes abort-time leakage but not the information in the full outcome word. Hiding that word would require a separate cryptographic or measurement design. We therefore use ``value-private unanimity'' only for the indistinguishability of $0^n$ and $1^n$.

\subsection{Local maximal mixing}

For honest Bell references, every traveling subsystem is locally $I/2$ and remains so under every user rotation, see equations~\eqref{eq:maxmixed} and~\eqref{eq:bell-local-invariance}. This restricts standalone single-point information available from the traveler while allowing globally distinct correlations after it is reunited with Alice's retained subsystem.

Section~\ref{sec:bell-boundary} shows why this property should not be conflated with the additive operator law. A pure qubit with $r_y=0$ can reproduce the same ideal rotational $\cos^2\Theta$ statistics, and private reference randomization can make an observer-averaged ensemble maximally mixed. Nevertheless, each realized pure carrier remains rank one. Bell local maximal mixing is therefore a resource-level property of the chosen implementation, not the mathematical origin of the ordered-angle theorem. It also does not, by itself, prove security against a coherent probe that interacts with the traveler at multiple route segments.

\subsection{Explicit boundaries}

This work does not establish:
\begin{itemize}[leftmargin=2em]
\item composable secrecy against arbitrary coherent or collective attacks,
\item fairness or security against malicious users that choose correlated settings, intentional angle deviations, or selective aborts (Section~\ref{sec:robustness} treats bounded calibration and readout error, not adversarial deviations),
\item security against Alice or infrastructure that probes the traveler between user transformations,
\item privacy of the nonunanimous Hamming-weight class from the $N/C$ or early-abort transcript,
\item multiphoton, higher-dimensional, out-of-mode, Trojan-horse, wavelength, timing, detector, or loss-manipulation security,
\item finite-key bounds for conference reconciliation and privacy amplification,
\item device independence or measurement-device independence.
\end{itemize}
Practical deployment therefore requires optical isolation, filtering, monitoring, and access control against injected probes, as in practical QKD countermeasure analyses~\cite{gisin2006trojan,jain2015trojan}.

\section{Photonic binary Bell readout}
\label{sec:photonic}

The physical comparator is inherited from the two-user rotational Bell construction rather than introduced here~\cite{lizama2026twouser}. Alice need not resolve all four Bell labels. With the convention of equation~\eqref{eq:bellref}, she applies
\begin{equation}
Q_{r_j}=P_{r_j\oplus 11},
\end{equation}
which maps $\ket{\beta_{r_j}}$ to the singlet $\ket{\Psi^-}$, up to phase, and maps its orthogonal Bell complement to the triplet subspace. The readout becomes
\begin{equation}
\Pi_s=\proj{\Psi^-},
\qquad
\Pi_t=I_4-\Pi_s.
\label{eq:singtrip}
\end{equation}

For indistinguishable photons incident on a balanced beam splitter, the antisymmetric singlet antibunches into separate outputs while symmetric triplet components bunch~\cite{hong1987hom,mitchell2003filter}. Photon-number-resolving detection, or an equivalent splitter-tree arrangement, is useful to distinguish genuine two-photon bunching from single-photon loss. The required operation is therefore a $1+3$ symmetry filter rather than a complete $1+1+1+1$ Bell analyzer. Complete linear-optical Bell discrimination has its own efficiency limits~\cite{calsamiglia2001bsm}.

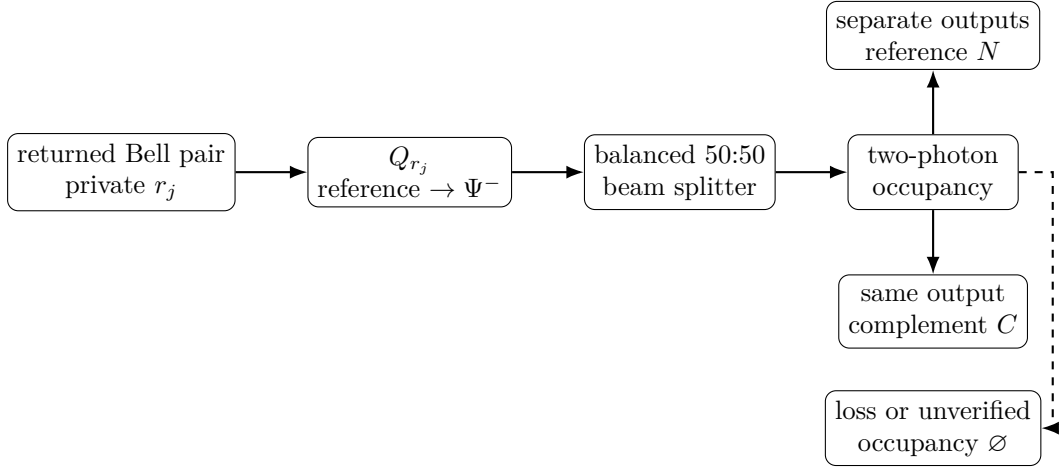
\begin{figure}[H]
\centering
\begin{tikzpicture}[
>=Latex,
node distance=0.75cm and 0.95cm,
box/.style={draw,rounded corners,minimum width=2.25cm,minimum height=0.76cm,align=center,font=\small},
arr/.style={->,thick}
]
\node[box] (ret) {returned Bell pair\\private $r_j$};
\node[box,right=of ret] (norm) {$Q_{r_j}$\\reference $\to\Psi^-$};
\node[box,right=of norm] (bs) {balanced 50:50\\beam splitter};
\node[box,right=of bs] (det) {two-photon\\occupancy};
\node[box,above=0.85cm of det] (sing) {separate outputs\\reference $N$};
\node[box,below=0.85cm of det] (trip) {same output\\complement $C$};
\node[box,below=0.55cm of trip] (null) {loss or unverified\\occupancy $\varnothing$};
\draw[arr] (ret)--(norm);
\draw[arr] (norm)--(bs);
\draw[arr] (bs)--(det);
\draw[arr] (det)--(sing);
\draw[arr] (det)--(trip);
\draw[arr,dashed] (det.east) -- ++(0.45,0) |- (null.east);
\end{tikzpicture}
\caption{Inherited binary Bell comparator reused at every word position. Verified separate-output and same-output two-photon records give $N$ and $C$, loss or unverified occupancy gives $\varnothing$. Only the public rotation magnitude changes between trials.}
\label{fig:receiver}
\end{figure}

Loss or invalid two-photon occupancy is a third erasure outcome and is not identified with $N$ or $C$. The algebraic formulas are conditioned on valid binary classifications. Setting-dependent or adversarially controlled loss remains outside the security claim.

\section{Discussion}

The main result is a code-design statement rather than a new use of Bell states or rotations by themselves. Equal-angle repetition raises one likelihood to a higher power. Ordered-angle coding changes the angular question from position to position. For $n=2^r$, the dyadic coefficients resolve the $2$-adic valuation classes of all mixed Hamming weights, and the minimality theorem shows that $\log_2n$ positions are necessary and sufficient within the deterministic-unanimity same-axis pulse family. The restriction to this family is essential. We do not claim an impossibility result for arbitrary real-angle schedules, adaptive measurements, or higher-dimensional resources.

The robustness analysis connects that exact theorem to device requirements. A bounded additive per-user angle offset lifts an ideal rejecting zero to $\sin^2\Delta_j$, and the conservative worst-case accumulation scales as $n\delta_{\max}$. The mathematical gain therefore survives small calibration and readout errors, but it is accompanied by a clear $1/n$ tightening of the per-user angular tolerance at fixed false-acceptance target. This is a directly testable requirement for a polarization-rotation implementation.

Bell entanglement has a distinct role. The same additive rotation law and ideal $\cos^2\Theta$ statistics can be reproduced with suitable pure-qubit references. The Bell implementation is retained because the traveler is locally maximally mixed in every honest realization, while the relation is recovered only through the global correlation with Alice's retained reference. This improves the resource-level information geometry but does not close the security problem. Current QCKA research explicitly treats finite-key coherent attacks, restricted users, detector trust, and source trust~\cite{krawec2025cad,barreiro2025sqcka,yang2024mdiqcka,hua2025siqcka}. A comparable security layer for serial Loop-Back access remains future work.

Finally, ordered angles improve verification, not the prior probability of spontaneous unanimity. For independent unbiased user bits that probability remains $2^{1-n}$. The optional conference-key interpretation should therefore be read as raw-candidate generation after an exact relation pass, not as a high-rate or composably secure QCKA protocol. The most immediate extensions are optimized finite-budget words for composite $n$, experimental characterization of angle and visibility errors, and a security proof that includes intermediate probing, adversarial loss, reconciliation, and privacy amplification.

\section{Conclusion}

Ordered-angle coding converts serial same-axis rotations into a position-dependent multiuser relation test. For $\alpha_j=q_j\pi/(2n)$, the mixed-weight imitation law is $M_{\mathbf q}(w)=\prod_j\cos^2(\pi q_jw/n)$. Within this deterministic-unanimity pulse family, a finite perfect word exists exactly when $n$ is a power of two, and the dyadic word $(1,2,4,\ldots,2^{r-1})$ is pulse-minimal with $m_{\min}=r=\log_2n$. Odd-prime user counts instead admit balanced statistical words, while composite counts with odd factors motivate finite-budget optimization.

The ideal zeros produced by an exact word remain controlled under small implementation errors. Per-user additive angle offsets accumulate into $\Delta_j=\sum_i\delta_{ij}$ and lift a rejecting position to $\sin^2\Delta_j$. Combined with a symmetric binary classification model, this gives explicit false-acceptance and true-acceptance bounds and a direct calibration scaling with user count.

The codeword theorem does not rely on entanglement alone. Bell references are used because the traveler remains locally maximally mixed in every honest trial, separating the local carrier from the globally readable relation. An accepted exact word can certify $[b_1=\cdots=b_n]$ and provide a raw conference-key candidate, but composable conference-key security, malicious-user fairness, intermediate active probing, adversarial loss, and finite-key post-processing remain outside the present claim.

\section*{Data and code availability}
All numerical data underlying the figures and tables are generated from the closed-form expressions in the article. The verification script and its output are supplied with the source material. No experimental data are reported.

\section*{Competing interests}
The author declares no competing interests.

\section*{Author contributions}
L.A.L.-P. conceived the protocol family, developed the ordered-angle construction and analytical results, and wrote the manuscript.

\appendix

\section{Derivation of the pulse-match probability}

For trial $j$,
\[
\alpha_j=\frac{q_j\pi}{2n},
\qquad
S=n-2w.
\]
The complement probability is
\[
P(C|w,q_j)
=
\sin^2\left[
\frac{q_j\pi}{2}
-\frac{q_j\pi w}{n}
\right].
\]
If $q_j$ is odd, unanimity expects $C$ and
\[
P(\text{match}|w,q_j)
=
\cos^2\left(\frac{q_j\pi w}{n}\right).
\]
If $q_j$ is even, unanimity expects $N$. In that case
\[
P(C|w,q_j)
=
\sin^2\left(\frac{q_j\pi w}{n}\right),
\]
so
\[
P(N|w,q_j)
=
\cos^2\left(\frac{q_j\pi w}{n}\right).
\]
Thus the same match formula applies for all integer $q_j$.

\section{Trigonometric identity for the prime word}

For odd integer $p$,
\[
\prod_{q=1}^{p-1}
\sin\left(\frac{\pi q}{p}\right)
=
\frac{p}{2^{p-1}}.
\]
Pairing $q$ with $p-q$ gives
\[
\prod_{q=1}^{(p-1)/2}
\sin^2\left(\frac{\pi q}{p}\right)
=
\frac{p}{2^{p-1}}.
\]
A complementary standard identity is
\[
\prod_{q=1}^{(p-1)/2}
\cos\left(\frac{\pi q}{p}\right)
=
2^{-(p-1)/2},
\]
which follows, for example, from
\[
\sin(2x)=2\sin x\cos x
\]
and the permutation of the nonzero residues under multiplication by $2$ modulo an odd prime. Squaring gives the factor $2^{-(p-1)}$ used in Theorem~\ref{thm:primeword}.

\section{General early-abort enumeration}
\label{sec:earlyabort-app}

For a fixed mixed weight $w$, define cumulative survival
\[
R_0(w)=1,\qquad
R_j(w)=\prod_{\ell=1}^{j}\mu_\ell(w).
\]
An input block reaches trial $j$ precisely when the first $j-1$ trials matched the unanimous signature, so
\[
\bar m(w|\mathbf q)
=
\sum_{j=1}^{m}R_{j-1}(w).
\]
For independently uniform input bits, the expected number of valid Bell trials used by an early-abort implementation is
\[
\bar m(\mathbf q)=2^{-n}\left[2m+\sum_{w=1}^{n-1}\binom{n}{w}\,\bar m(w\mid\mathbf q)\right].
\]
The first term accounts for the two unanimous strings, which always reach all $m$ positions. The remaining terms weight each mixed Hamming class by its binomial multiplicity.

\bibliographystyle{unsrtnat}
\bibliography{references}

\end{document}